\documentclass[a4paper,psamsfonts,reqno]{amsart}

\usepackage{mathpazo}
\usepackage{mathbbol,bm,amsmath,amssymb}
\usepackage{graphicx}
\usepackage{color}
\usepackage{algpseudocode}
\usepackage{algorithm}
\usepackage{url}
\usepackage{cite}
\usepackage{natbib}

\DeclareMathOperator*{\argmax}{argmax}
\DeclareMathOperator*{\argmin}{argmin}
\DeclareMathOperator{\sgn}{sgn}

\newtheorem{thm}{Theorem}
\newtheorem{prop}{Proposition}
\newtheorem{lem}{Lemma}

\newcommand{\R}{\mathbb{R}}
\newcommand{\Se}{\mathcal{S}}
\newcommand{\gev}{\succcurlyeq}
\newcommand{\lev}{\preccurlyeq}
\newcommand{\grad}{\nabla}
\newcommand{\bone}{\hat\beta^{(1,k)}}
\newcommand{\btwo}{\hat\beta^{(2,k)}}

\begin{document}

\title[Natural coordinate descent for L1 regression]{Natural
  coordinate descent algorithm for L1-penalised regression in
  generalised linear models}

\author{Tom Michoel}

\address{The Roslin Institute, The University of Edinburgh, Easter
  Bush, Midlothian, EH25 9RG, Scotland, UK}

\email{tom.michoel@roslin.ed.ac.uk}

\thanks{Research supported by Roslin Institute Strategic Grant funding
  from the BBSRC}

\begin{abstract}
  The problem of finding the maximum likelihood estimates for the
  regression coefficients in generalised linear models with an
  $\ell_1$ sparsity penalty is shown to be equivalent to minimising
  the unpenalised maximum log-likelihood function over a box with
  boundary defined by the $\ell_1$-penalty parameter. In one-parameter
  models or when a single coefficient is estimated at a time, this
  result implies a generic soft-thresholding mechanism which leads to
  a novel coordinate descent algorithm for generalised linear models
  that is entirely described in terms of the natural formulation of
  the model and is guaranteed to converge to the true optimum. A
  prototype implementation for logistic regression tested on two
  large-scale cancer gene expression datasets shows that this
  algorithm is efficient, particularly so when a solution is computed
  at set values of the $\ell_1$-penalty parameter as opposed to along
  a regularisation path. Source code and test data are available from
  \url{http://tmichoel.github.io/glmnat/}.
\end{abstract}

\maketitle

\section{Introduction}
\label{sec:introduction}

In high-dimensional regression problems where the number of potential
model parameters greatly exceeds the number of training samples, the
use of an $\ell_1$ penalty which augments standard objective functions
with a term that sums the absolute effect sizes of all parameters in
the model has emerged as a hugely successful and intensively studied
variable selection technique, particularly for the ordinary least
squares (OLS) problem (e.g. \citep{tibshirani1996regression,
  osborne2000lasso, osborne2000new, efron2004least,
  zou2005regularization, johnstone2009statistical,
  friedman2010regularization, elghaoui2012,tibshirani2012strong,
  tibshirani2013lasso}).  Generalised linear models (GLMs) relax the
implicit OLS assumption that the response variable is normally
distributed and can be applied to, for instance, binomially
distributed binary outcome data or poisson distributed count data
\citep{nelder1972generalized}. However, the most popular and efficient
algorithm for $\ell_1$-penalised regression in GLMs uses a quadratic
approximation to the log-likelihood function to map the problem back
to an OLS problem and although it works well in practice, it is not
guaranteed to converge to the optimal solution
\citep{friedman2010regularization}.  Here it is shown that calculating
the maximum likelihood coefficient estimates for $\ell_1$-penalised
regression in generalised linear models can be done via a coordinate
descent algorithm consisting of successive soft-thresholding
operations on the \emph{unpenalised} maximum log-likelihood function
without requiring an intermediate OLS approximation. Because this
algorithm can be expressed entirely in terms of the natural
formulation of the GLM, it is proposed to call it the \emph{natural
  coordinate descent algorithm}.

To make these statements precise, let us start by introducing a
response variable $Y\in\R$ and predictor vector $X\in\R^p$. It is
assumed that $Y$ has a probability distribution from the exponential
family, written in canonical form as
\begin{equation*}
  p(y\mid\eta,\phi) = h(y,\phi) \exp\Bigl(\alpha(\phi)\bigl\{y\eta -
    A(\eta)\bigr\}\Bigr),
\end{equation*}
where $\eta\in\R$ is the natural parameter of the distribution, $\phi$
is a dispersion parameter and $h$, $\alpha>0$ and $A$ convex are known
functions.  The expectation value of $Y$ is a function of the natural
parameter, $E(Y)=A'(\eta)$, and linked to the predictor variables by
the assumption of a linear relation $\eta=X^T\beta$, where
$\beta\in\R^p$ is the vector of regression coefficients. It is tacitly
assumed that $X_1\equiv 1$ such that $\beta_1$ represents the
intercept parameter. Suppose now that we have $n$ observation pairs
$(x_i,y_i)$ (with $x_{i1}=1$ fixed for all $i$).  The minus
log-likelihood of the observations for a given set of regression
coefficients $\beta$ under the GLM is given by
\begin{equation}\label{eq:7}
  H(\beta) = \frac1n\sum_{i=1}^n  A(x_i^T\beta) - y_i(x_i^T\beta) \equiv
  U(\beta) - w^T\beta,
\end{equation}
where any terms not involving $\beta$ have been omitted, $U(\beta) =
\frac1n\sum_{i=1}^n A(x_i^T\beta)$ is a convex function,
$w=\frac1n\sum_{i=1}^n y_i x_i \in\R^p$, and the dependence of $U$ and
$w$ on the data $(x_i,y_i)$ has been suppressed for notational
simplicity. In the penalised regression setting, this cost function is
augmented with $\ell_1$ and $\ell_2$ penalty terms to achieve
regularity and sparsity of the minimum-energy solution, i.e. $H$ is
replaced by
\begin{equation}\label{eq:14}
  H(\beta) = U(\beta) - w^T\beta + \lambda \|\beta\|_2^2 + \mu \|\beta\|_1,
\end{equation}
where $\|\beta\|_2 = (\sum_{j=1}^p |\beta_j|^2)^{\frac12}$ and
$\|\beta\|_1 = \sum_{j=1}^p |\beta_j|$ are the $\ell_2$ and $\ell_1$
norm, respectively, and $\lambda$ and $\mu$ are positive
constants. The $\ell_2$ term merely adds a quadratic function to $U$
which serves to make its Hessian matrix non-singular and it will not
need to be treated explicitly in our analysis.  Furthermore a slight
generalisation is made where instead of a fixed parameter $\mu$, a
vector of predictor-specific penalty parameters $\mu_j$ is used. This
allows for instance to account for the usual situation where the
intercept coefficient is unpenalised ($\mu_1=0$).  The problem we are
interested in is thus to find
\begin{equation}\label{eq:3}
  \hat \beta = \argmin_{\beta\in\R^p} H(\beta),
\end{equation}
with $H$ a function of the form
\begin{equation}\label{eq:13}
  H(\beta) = U(\beta) - w^T\beta + \sum_{j=1}^p \mu_j |\beta_j|,
\end{equation}
where $U:\R^p\to\R$ is a smooth convex function, $w\in\R^p$ is an
arbitrary vector and $\mu\in\R^p$, $\mu\gev0$ is a vector of
non-negative parameters. The notation $u\gev v$ is used to indicate
that $u_j\geq v_j$ for all $j$ and likewise the notation $u\cdot v$
will be used to indicate elementwise multiplication, i.e. $(u\cdot
v)_j=u_jv_j$. The maximum of the \emph{unpenalised} log-likelihood,
considered as a function of $w$, is of course the Legendre transform
of the convex function $U$,
\begin{equation*}
  L(w) = \max_{\beta\in\R^p} \Bigl\{ w^T\beta -  U(\beta) \Bigr\},
\end{equation*}
and the unpenalised regression coefficients satisfy
\begin{align*}
  \hat\beta_0(w) = \argmax_{\beta\in\R^p} \Bigl\{ w^T\beta -  U(\beta)
  \Bigr\} = \grad L(w),
\end{align*}
where $\grad$ is the usual gradient operator (see Lemma
\ref{lem:legendre} in Appendix \ref{sec:proof-theorem-app}).  This
leads to the following key result:
\begin{thm}\label{thm:main}
  The solution $\hat\beta(w,\mu)$ of
  \begin{equation}\label{eq:1}
    \hat\beta(w,\mu) = \argmin_{\beta\in\R^p} \Bigl\{ U(\beta) -
    w^T\beta + \sum_{j=1}^p \mu_j|\beta_j| \Bigr\}
  \end{equation}
  is given by
  \begin{align*}
    \hat\beta(w,\mu) = \hat\beta_0\bigl(\hat u(w,\mu)\bigr) = \grad
    L\bigl(\hat u(w,\mu)\bigr),
  \end{align*}
  where $\hat u(w,\mu)$ is the solution of the constrained
  convex optimisation problem
  \begin{equation}\label{eq:2}
    \hat u(w,\mu) = \argmin_{\{u\in\R^p\colon  |u-w|\lev \mu\}} L(u).
  \end{equation}
  Furthermore the sparsity patterns of $\hat\beta$ and $\hat u - w
  +\sgn(\hat\beta)\cdot\mu$ are complementary,
  \begin{align*}
    \hat\beta_j(w,\mu) \neq 0 \Leftrightarrow \hat u_j(w,\mu) = w_j
    -\sgn(\hat\beta_j)\mu_j.
  \end{align*}
\end{thm}
The proof of this Theorem consists of an application of Fenchel's
duality theorem and is provided in Appendix
\ref{sec:proof-theorem-app}. Two special cases of Theorem
\ref{thm:main} merit attention. Firstly, in the case of lasso or
elastic net penalised linear regression,
$U(\beta)=\frac12\beta^T C\beta$ is a quadratic function of $\beta$,
with $C\in\R^{p\times p}$ a positive definite matrix, such that
$L(w)=\frac12 w^TC^{-1}w$ and $\grad L(w)=C^{-1}w$.  If furthermore
$C$ is diagonal, with diagonal elements $c_j>0$, then eq. \eqref{eq:2}
reduces to the $p$ independent problems
\begin{equation*}
  \hat u_j=\argmin_{\{u\in\R\colon |u-w_j|\leq\mu_j\}} \tfrac{1}{c_j}u_j^2,
\end{equation*}
with solution
\begin{equation*}
  \hat u_j =
  \begin{cases}
    \sgn(w_j) (|w_j|-\mu_j) & \text{if } |w_j|>\mu_j\\
    0 & \text{otherwise}
  \end{cases}
\end{equation*}
and $\hat\beta_j(w,\mu)=\frac1{c_j}\hat u_j$. This is the well-known
analytic solution of the lasso with uncorrelated predictors
\citep{tibshirani1996regression}, which forms the basis for
numerically solving the case of arbitrary $C$ as well
\citep{friedman2010regularization}.  Secondly, in the case of
penalised covariance matrix estimation, $U(\Theta)=-\log\det \Theta$
for non-negative definite matrices $\Theta\in\R^{p\times p}$, and
$L(W)=\log\det(-W)^{-1}$ for $W\in\R^{p\times p}$ negative definite
(and $+\infty$ otherwise) \citep[\S
3.3.1]{boyd2004convex}. Eq. \eqref{eq:2} is then exactly the dual
problem studied by \citet{banerjee2008model}.

\section{Natural coordinate descent algorithm}
\label{sec:natur-coord-desc}

\subsection{Exact algorithm}
\label{sec:natur-coord-desc-1}

It is well-known that a cyclic coordinate descent algorithm for the
$\ell_1$-penalised optimisation problem in eq. \eqref{eq:1} converges
\citep{tseng2001convergence}.  When only one variable is optimised at
a time, keeping all others fixed, the equivalent variational problem
in eq. \eqref{eq:2} reduces to a remarkably simple soft-thresholding
mechanism illustrated in Figure \ref{fig:Thm-1D}. More precisely, let
$U(\beta)$ be a smooth convex function of a single variable
$\beta\in\R$, $w_0=\argmin_{u\in\R}L(u)$ and $\sigma=\sgn(w-w_0)$. The
solution of the one-variable optimisation problem
\begin{align*}
  \hat u = \argmin_{\{u\in\R\colon |u-w|\leq\mu\}} L(u),
\end{align*}
with $\mu\geq0$, can be expressed as follows. If $|w-w_{0}|\leq\mu$ then
$\hat u=w_{0}$ and hence $\hat \beta(w,\mu)=0$. Otherwise we must have
$\hat u=w-\sigma\mu$ and
$\hat\beta(w,\mu)=L'(w-\sigma\mu)=\hat\beta_{0}(w-\sigma\mu)$.
Hence the solution takes the form of a generalised `soft-thresholding'
\begin{equation*}
  \hat\beta(w,\mu) =
  \begin{cases}
    \hat\beta_{0}(w-\sigma\mu) & |w-w_{0}|>\mu\\
    0 & |w-w_{0}|\leq\mu
  \end{cases},
\end{equation*}
see also Figure \ref{fig:Thm-1D}. In other words, compared to the
multivariate problem in Theorem \ref{thm:main} where there remains
ambiguity about the signs $\sgn(\hat\beta_j)$, in the one-variable
case the sign is uniquely determined by the relative position of $w$
and $w_0$.

\begin{figure}
  \centering
  \includegraphics[width=\linewidth]{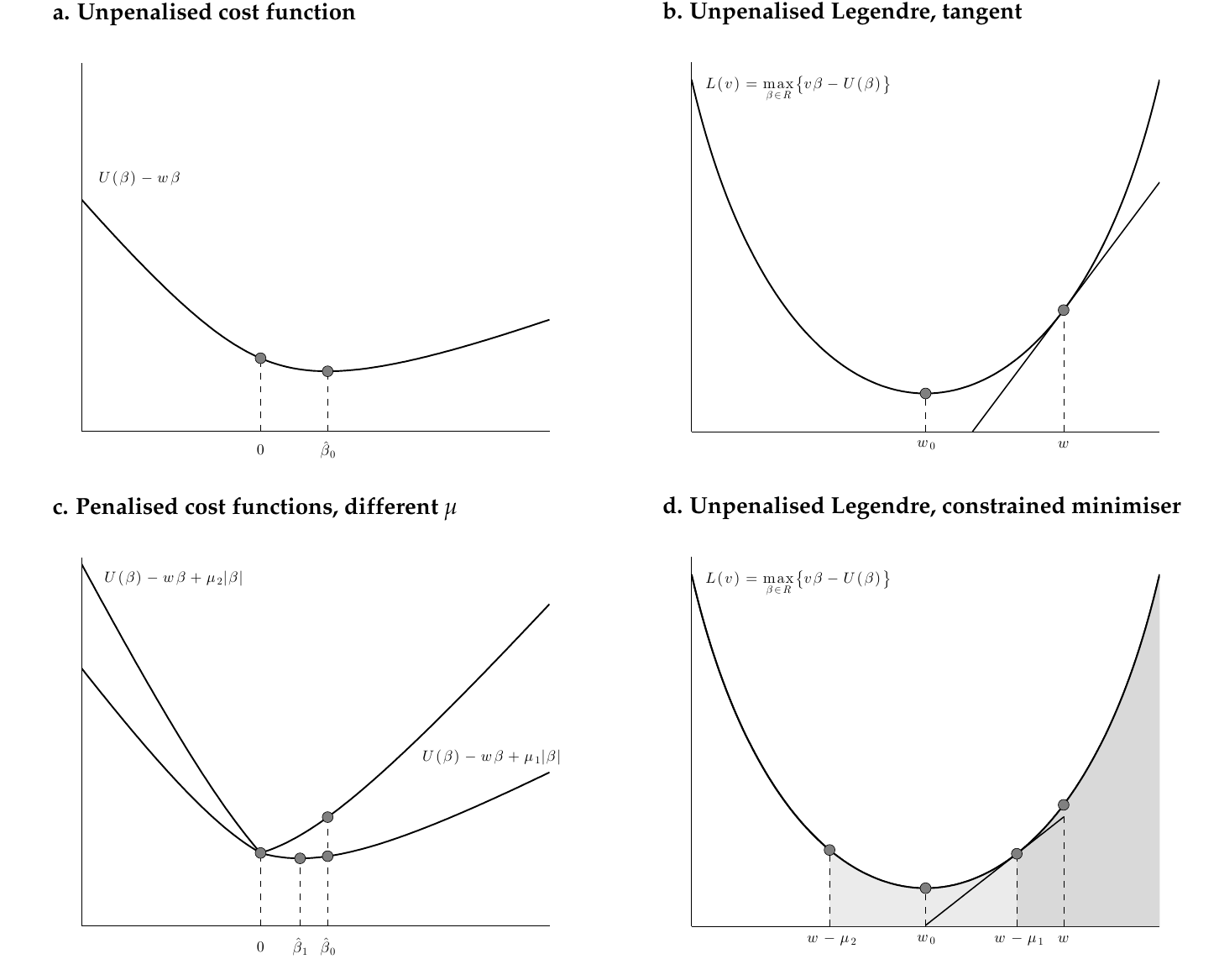}
  \caption{Illustration of Theorem \ref{thm:main} in one
    dimension. \textbf{a.} The unpenalised cost function
    $U(\beta)-w\beta$ is a convex function of $\beta$; the
    maximum-likelihood estimate $\hat\beta_0$ is its unique
    minimiser. \textbf{b.} The maximum-likelihood estimate is also
    equal to the slope of the tangent to the Legendre transform of $U$
    at $w$. \textbf{c.} Every value of the $\ell_1$ penalty parameter
    $\mu$ leads to a different cost function; for $\mu=\mu_1$
    sufficiently small, the maximum-likelihood estimate
    $\hat\beta_1<\hat\beta_0$ is non-zero while for sufficiently large
    $\mu=\mu_2$ it is exactly zero. \textbf{d.} The penalised problem
    can also be solved by minimising the \emph{unpenalised} Legendre
    transform over the interval $[w-\mu,w+\mu]$; for $w>w_0$ and
    $\mu_1<w-w_0$ the absolute minimiser of $L$ is not included in
    this interval such that the constrained minimiser is the boundary
    value $w-\mu_1$ and the the maximum-likelihood estimate
    $\hat\beta_1$ equals the slope of the tangent at $w-\mu_1$, while
    for $\mu_2\geq w-w_0>0$, the constrained minimiser is always the
    absolute minimiser which has a tangent with slope zero. Note that
    because $L$ is convex, the slope at $w-\mu_1$ is always smaller
    than the slope at $w$ (i.e. $\hat\beta_1<\hat\beta_0$). Similar
    reasoning applies when $w<w_0$.}
  \label{fig:Thm-1D}
\end{figure}

Numerically solving the unpenalised one-variable problem is usually
straightforward. First note that by assumption, $U$ is differentiable
and therefore it is itself the Legendre transform of $L$. Hence
\begin{equation*}
  w_{0} = \argmin_{u\in\R} L(u) = \Bigl.
  \argmax_{u\in\R}\bigl\{\beta u- L(u)\bigr\}\Bigr|_{\beta=0} =
  U'(0). 
\end{equation*}
Likewise, and assuming there exists no analytic expression for $L$,
$\hat\beta_{0}(w-\sigma\mu)$ can be found as the zero of the function
\begin{align*}
  f(\beta) = U'(\beta) - w + \sigma\mu.
\end{align*}
For $U$ convex, this is a monotonically increasing function of $\beta$
and conventional one-dimensional root-finding algorithms converge
quickly. 

The $p$-dimensional natural coordinate descent algorithm simply
consists of iteratively applying the above procedure to the
one-dimensional functions
\begin{equation*}
  U_j(\beta_j) = U(\hat\beta_1,\dots,
  \hat\beta_{j-1},\beta_j,\hat\beta_{j+1},\dots, \hat\beta_{p}),
\end{equation*}
where $\hat\beta\in\R^p$ are the current coefficient estimates, i.e.
\begin{equation}\label{eq:16}
  \hat\beta_j^{\text{(new)}} = 
  \begin{cases}
    \argmax_{\beta\in\R}\{(w_j-\sigma_j\mu_j)\beta - U_j(\beta)\}  &
    \text{ if } \bigl|w_j-w_{0,j}\bigr| > \mu_j\\
    0 & \text{ otherwise}
  \end{cases},
\end{equation}
where $w_{0,j} = U'_j(0)$ and $\sigma_j=\sgn(w_j-w_{0,j})$.

Standard techniques can be used to make the algorithm more efficient
by organising the calculations around the set of non-zero coefficients
\citep{friedman2010regularization}, that is, after every complete cycle
through all coordinates, the current set of non-zero coefficients is
updated until convergence before another complete cycle is run (see
pseudocode in  Appendix \ref{sec:code}).

\subsection{Quadratic approximation algorithm with exact thresholding}
\label{sec:quadr-appr-with}

An alternative method for updating $\hat\beta_j$ in the preceding
algorithm is to use a quadratic approximation to $U_j(\beta_j)$ around
the current estimate of $\hat\beta_j$ in eq. \eqref{eq:16}. This leads
to a linear approximation for $\hat\beta_j^{\text{(new)}} $, i.e. if
$|w_j-w_{0,j}| > \mu_j$, then
\begin{align}
  \hat\beta_j^{\text{(new)}} &=
  \argmax_{\beta\in\R}\{(w_j-\sigma_j\mu_j)\beta - U_j(\beta)\} \nonumber\\
  &\approx \argmax_{\beta\in\R}\bigl\{ (w_j-\sigma_j\mu_j) \beta -
  U_j(\hat\beta_j) - U_j'(\hat\beta_j) (\beta-\hat\beta_j) - \frac12
  U_j''(\hat\beta_j) (\beta-\hat\beta_j)^2 \bigl\} \nonumber\\
  &= \hat\beta_j + \frac{w_j-\sigma_j\mu_j -
    U_j'(\hat\beta_j)}{U_j''(\hat\beta_j)}.\label{eq:5}
\end{align}
This approximation differs from the standard quadratic approximation
\citep{friedman2010regularization} by the fact that it still uses the
\emph{exact} thresholding rule from \eqref{eq:16}. To be precise,
given current estimates $\hat\beta$, the standard approximation
updates the $j^{\text{th}}$ coordinate by minimizing the approximate
quadratic cost function
\begin{align*}
  \frac12 U_j''(\hat\beta_j) \beta_j^2 - \bigl[ w_j -
  U'_j(\hat\beta_j) + U_j''(\hat\beta_j) \hat\beta_j\bigr] \beta_j +
  \mu_j|\beta_j|,
\end{align*}
which has the solution
\begin{align}\label{eq:9}
  \hat\beta_j^{\text{(new)}}=
  \begin{cases}
    \hat\beta_j + \frac{w_j - U'_j(\hat\beta_j) - \sigma'_j \mu_j
    }{U_j''(\hat\beta_j) } & \text{ if } \bigl|w_j - U'_j(\hat\beta_j) +
    U_j''(\hat\beta_j) \hat\beta_j\bigr|>\mu_j\\
    0 & \text{ otherwise}
  \end{cases},
\end{align}
where
$\tilde\sigma_j = \sgn(w_j - U'_j(\hat\beta_j) + U_j''(\hat\beta_j)
\hat\beta_j)$.

Hence, compared to the exact coordinate update rule \eqref{eq:16}, the
standard algorithm not only uses a quadratic approximation to the cost
function, but also a linear approximation
\begin{equation*}
  w_{0,j} = U'_j(0) \approx U'_j(\hat\beta_j) - U_j''(\hat\beta_j)
  \hat\beta_j.
\end{equation*}
 
The following result shows that, under certain conditions, the
approximate and exact thresholding will return the same result:
\begin{prop}\label{prop:threshold}
  Let $U(\beta)$ be a smooth convex
    function of a single variable $\beta\in\R$, and let $\hat\beta$ be
    the solution of
  \begin{align*}
    \hat\beta = \argmin_{\beta\in\R} \bigl\{ U(\beta) - w\beta + \mu
    |\beta|\bigr\}, 
  \end{align*}
  with $w\in\R$ and $\mu>0$. Denote $w_0=U'(0)$ and
  $\tilde w_0 = U'(\hat\beta) - U''(\hat\beta) \hat\beta$. Then
  \begin{align*}
    |w-w_0| > \mu \Leftrightarrow |w-\tilde w_0|>\mu.
  \end{align*}
\end{prop}
The proof of this proposition can be found in Appendix
\ref{sec:proof-theorem-app}. Note that in the coordinate descent
algorithms the single-coordinate functions $U$ change from step to
step, that $\tilde w_0$ is calculated on the \emph{current} instead of
the \emph{new} solution, and that, in the quadratic approximation
algorithm, both the current and new solutions are only approximate
minimisers. Hence this result only shows that if all these errors are
sufficiently small, then both thresholding rules will agree.

\section{Numerical experiments}
\label{sec:numer-exper}

I implemented the natural coordinate descent algorithm for logistic
regression in C with a Matlab interface (source code available from
\url{http://tmichoel.github.io/glmnat/}).  The penalised cost function for
$\beta\in\R^p$ in this case is given by
\begin{align*}
  H(\beta) = U(\beta) - w^T\beta + \mu \sum_{j=2}^p |\beta_j|,
\end{align*}
where
\begin{align*}
  U(\beta) = \frac1n\sum_{i=1}^n \log\bigl(1+e^{x_i^T\beta}\bigr)\;,
  \quad  w = \frac1n\sum_{i=1}^n y_i x_i^T,
\end{align*}
and $(x_i\in\R^p,y_i\in\{0,1\})$, $i=1,\dots,n$ are the
observations. Recall from Section \ref{sec:introduction} that
$\beta_1$ is regarded as the (unpenalised) intercept parameter and
therefore a fixed value of one ($x_{i1}=1$) is added to every
observation. As convergence criterion I used
$\max_{j=1,\dots,p}|\hat\beta_j^{(\text{new})} -
\hat\beta_j^{(\text{old})}|<\epsilon$, where $\epsilon>0$ is a fixed
parameter. The difference is calculated at every iteration step when a
single coefficient is updated and the maximum is taken over a full
iteration after all, resp. all active, coefficients have been updated
once.

To test the algorithm I used gene expression levels for 17,814 genes
in 540 breast cancer samples (BRCA dataset) \citep{tcgabreast2012} and
20,531 genes in 266 colon and rectal cancer samples (COAD dataset)
\citep{tcgacolon2012} as predictors for estrogen receptor status
(BRCA) and early-late tumor stage (COAD), respectively (see Appendix
\ref{sec:cancer-genome-atlas} for details, processed data available
from \url{http://tmichoel.github.io/glmnat/}). I compared the
implementation of the natural coordinate descent algorithm against
\texttt{glmnet} (version dated 30 Aug 2013) \citep{qian2013}, a
Fortran-based implementation for Matlab of the coordinate descent
algorithm for penalised regression in generalised linear models
proposed by \citet{friedman2010regularization}, which was found to be
the most efficient in a comparison to various other softwares by the
original authors \citep{friedman2010regularization} as well as in an
independent study \citep{yuan2010comparison}. All analyses were run on
a laptop with 2.7~GHz processor and 8~GB RAM using Matlab v8.2.0.701
(R2013b).

Following \citet{friedman2010regularization}, I considered a geometric
path of regularisation parameters
\begin{equation}\label{eq:11}
  \mu^{(k)}=\frac{\max_{j=2,\dots,p}|w_j|}{m^{\frac{k-1}{m-1}}},
\end{equation}
where $m=100$, and $k=1,\dots,m$, corresponding to the default choice
in \texttt{glmnet}. Note that $\mu^{(1)}$ is the smallest penalty that
yields a solution where only the intercept parameter is non-zero.
Such a path of parameters evenly spaced on log-scale typically
corresponds to models with a linearly increasing number of non-zero
coefficients \citep{friedman2010regularization}.  To compare the
output of two different algorithms over the entire regularisation
path, I considered the maximum relative score difference
\begin{align*}
  \max_{k=1,\dots,m}\frac{H(\bone) - H(\btwo)}{H(\bone)},
\end{align*}
where $\bone$ and $\btwo$ are the coefficient estimates obtained by
the respective algorithms for the $k$th penalty parameter.

A critical issue when comparing algorithm runtimes is to match
convergence threshold settings.  Figure \ref{fig:comp}a shows the
runtimes of the exact natural coordinate descent algorithm (using
eq. \eqref{eq:16}) and its quadratic approximation (using
eq. \eqref{eq:5}), and their maximum relative score difference for a
range of values of the convergence threshold $\epsilon$. The quadratic
approximation algorithm is about twice as fast as the exact algorithm
and, as expected, both return numerically identical results within the
accepted tolerance levels. For subsequent analyses only the quadratic
approximation algorithm was used.  Because \texttt{glmnet} uses a
different convergence criterion than the one used here, I ran the
natural coordinate descent algorithm with a range of values for
$\epsilon$ and calculated the maximum relative score difference over
the entire regularisation path with respect to the output of
\texttt{glmnet} with default settings. Figure \ref{fig:comp}b shows
that there is a dataset-dependent value for $\epsilon$ where this
difference is minimised and that the minimum difference is within the
range observed when running \texttt{glmnet} with randomly permuted
order of predictors. These minimising values
$\epsilon_{\text{BRCA}}=6.3\times 10^{-4}$ and
$\epsilon_{\text{COAD}}=2.0\times 10^{-3}$ were used for the
subsequent comparisons.

First, I compared the natural coordinate descent algorithms with exact
and approximate thresholding rules (cf. eqs. \eqref{eq:16} and
\eqref{eq:9}). For both datasets and all penalty parameter values, no
differences were found between the two rules during the entire course
of the algorithm, indicating that the error terms discussed after
Proposition \ref{prop:threshold} are indeed sufficiently small in
practice. Since there is as yet no analytical proof extending
Proposition \ref{prop:threshold} to the algorithmic setting, the exact
thresholding rule was used for all subsequent analyses.

\begin{figure}
  \centering
  \includegraphics[width=\linewidth]{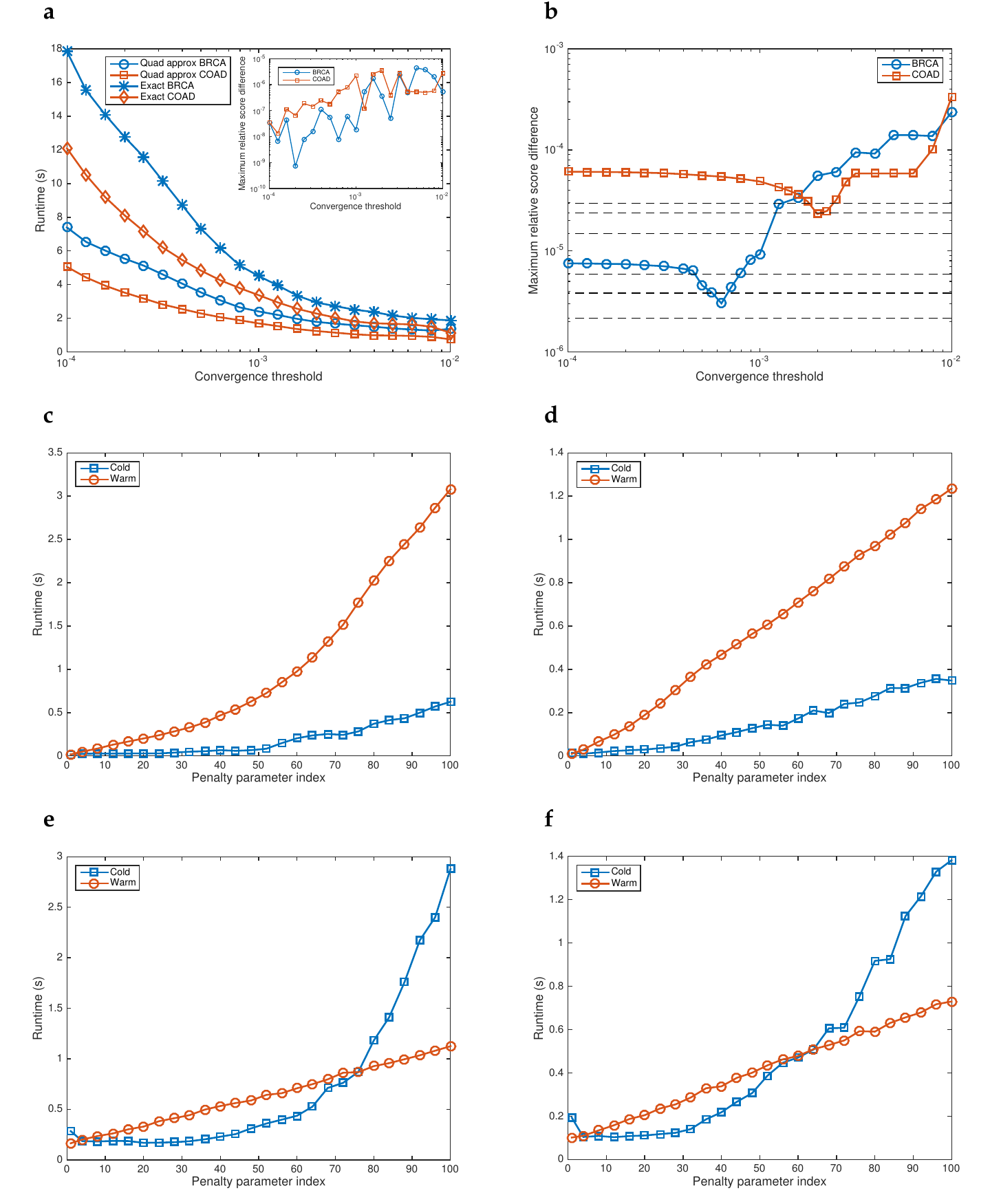}
  \caption{Numerical evaluation of the natural coordinate descent
    algorithm. \textbf{a.} Runtime in seconds of the quadratic
    approximation algorithmn on the BRCA (blue circles) and COAD (red
    squares) dataset and of the exact algorithm on the BRCA (blue
    asterisks) and COAD (red diamonds) dataset vs. convergence
    threshold parameter $\epsilon$. The inset shows the maximum
    relative score difference between both algorithms for the same
    convergence thresholds. \textbf{b.} Maximum relative score
    difference between the natural coordinate descent algorithm and
    \texttt{glmnet} vs. convergence threshold parameter $\epsilon$ for
    the BRCA (blue circles) and COAD (red squares) dataset. The
    horizontal lines indicate the minimum, mean and maximum of the
    relative score difference over 10 comparisons between the original
    \texttt{glmnet} result and \texttt{glmnet} applied to data with
    randomly permuted order of predictors.  \textbf{c,d.} Runtime in
    seconds of the natural coordinate descent algorithm with cold
    (blue squares) and warm (red circles) starts on the BRCA
    (\textbf{c}) and COAD (\textbf{d}) dataset vs. index $k$ of the
    penalty parameter vector. \textbf{e,f.} Runtime in seconds of
    \texttt{glmnet} with cold (blue squares) and warm (red circles)
    starts on the BRCA (\textbf{e}) and COAD (\textbf{f}) dataset
    vs. index $k$ of the penalty parameter vector. See main text for
    details.}
  \label{fig:comp}
\end{figure}

Next, I compared the natural coordinate descent algorithm to
\texttt{glmnet} considering both ``cold'' and ``warm'' starts. For the
cold starts, the solution for the $k^{\text{th}}$ penalty parameter
value, $\hat\beta(\mu^{(k)})$, was calculated starting from the
initial vector $\hat\beta=0$. For the warm starts,
$\hat\beta(\mu^{(k)})$ was calculated along the path of penalty
parameters $\mu^{(1)}, \mu^{(2)}, \dots, \mu^{(k)}$, each time using
$\hat\beta(\mu^{(l)})$ as the initial vector for the calculation of
$\hat\beta(\mu^{(l+1)})$. This scenario was designed to answer the
question: If a solution $\hat\beta(\mu)$ is sought for some fixed
value of $\mu<\mu^{(1)}$, is it best to run the coordinate descent
algorithm once starting from the initial vector $\hat\beta=0$ (cold
start), or to run the coordinate descent algorithm multiple times
along a regularization path, each time with an initial vector that
should be close to the next solution (warm start)? Clearly, if a
solution is needed for all values of a regularization path it is
always better to run through the path once using warm starts at each
step.

For \texttt{glmnet}, there is a clear advantage to using warm starts
and, as also observed by \citet{friedman2010regularization}, for
smaller values of $\mu$, it can be faster to compute along a
regularisation path down to $\mu$ than to compute the solution at
$\mu$ directly (Figure \ref{fig:comp}e,f). In contrast, the natural
coordinate descent algorithm is much less sensitive to the use of warm
starts (i.e. to the choice of initial vector for $\hat\beta$) and it
is considerably faster than \texttt{glmnet} when calculating solutions
at single penalty parameter values (Figure \ref{fig:comp}c,d). 

To investigate whether this qualitative difference between both
algorithms is a general property, the following process was repeated
1,000 times: a gene was randomly selected from the BRCA dataset, a
binary response variable was defined from the sign of its expression
level, and penalised logistic regression with penalty parameter
$\mu^{(90)}$ (cf. eq. \eqref{eq:11}) was performed using 5,000
randomly selected genes as predictors, using both cold start (with
initial vector $\hat\beta=0)$ and warm start (along the regularization
path $\mu^{(1)},\dots,\mu^{(90)}$); the response gene was constrained
to have at least $30\%$ samples of either sign and the predictor genes
were constrained to not contain the response gene.  This scheme
ensured that datasets with sufficient variability in the correlation
structure among the predictor variables and between the predictor and
the response variables were represented among the test cases. As
expected, total runtime correlated well with the size of the model,
defined here as the number of predictors with $|\beta_j|>10^{-3}$,
more strongly so for the natural coordinate descent algorithm
(Pearson's $\rho_{\text{cold}}=0.92$, $\rho_{\text{warm}}=0.91$) than
for \texttt{glmnet} ($\rho_{\text{cold}}=0.79$,
$\rho_{\text{warm}}=0.84$).  Consistent with these high linear
correlations the difference in speed (runtime$^{-1}$) between cold and
warm start was inversely proportional to model size (Spearman's
$\rho=-0.90$; Figure \ref{fig:speed}a). Furthermore, cold start
outperformed warm start ($v_{\text{cold}}>v_{\text{warm}}$) in all
1,000 datasets. For \texttt{glmnet} the opposite was true: warm start
always outperformed cold start. However the speed difference
($v_{\text{cold}}-v_{\text{warm}}$) did not correlate as strongly with
model size (Spearman's $\rho=0.53$; Figure \ref{fig:speed}b; note that
the opposite sign of the correlation coefficient is merely due to
the opposite sign of the speed differences).

This consistent qualitative difference between both algorithms with
respect to the choice of initial vector was unexpected in view of the
results in Section \ref{sec:quadr-appr-with}. Upon closer inspection,
it was revealed that the natural coordinate descent algorithm uses a
scheme whereby the parameters of the quadratic approximation for
coordinate $j$ (i.e., $U'_j(\hat\beta_j)$ and $U''_j(\hat\beta_j)$)
are updated whenever there is a change in $\hat\beta_{j'}$ for some
$j'\neq j$. In contrast, \texttt{glmnet} uses two separate loops,
called ``middle'' and ``inner'' loop in
\citep{friedman2010regularization}. In the middle loop, the quadratic
approximation to $U$ at the current solution $\hat\beta^*$ is
calculated, i.e.
\begin{equation}\label{eq:10}
  U(\beta) \approx U(\hat\beta^*) + \sum_{j=1}^p G_j
  (\beta_j-\hat\beta^*_j) + \sum_{j,k=1}^p H_{jk}
  (\beta_j-\hat\beta^*_j) (\beta_k-\hat\beta^*_k),
\end{equation}
where
\begin{align*}
  G_j = \frac{\partial U}{\partial \beta_j}\biggr|_{\hat\beta^*} \quad
  \text{and} \quad
  H_{jk} = \frac{\partial^2 U}{\partial \beta_j \partial
  \beta_k}\biggr|_{\hat\beta^*}.
\end{align*}
In the inner loop, a penalised least squares coordinate descent
algorithm is run until convergence using the approximation
\eqref{eq:10}, i.e. keeping the values of $G$ and $H$ fixed. A poor
choice of initial vector will therefore result in values of $G$ and
$H$ that are far from optimal, and running a coordinate descent
algorithm until convergence without updating these values would
therefore result in a loss of efficiency. It is therefore plausible
that the continuous updating of the quadratic approximation parameters
in the natural coordinate descent algorithm explains its robustness
with respect to the choice of initial vector.

\begin{figure}
  \centering
  \includegraphics[width=\linewidth]{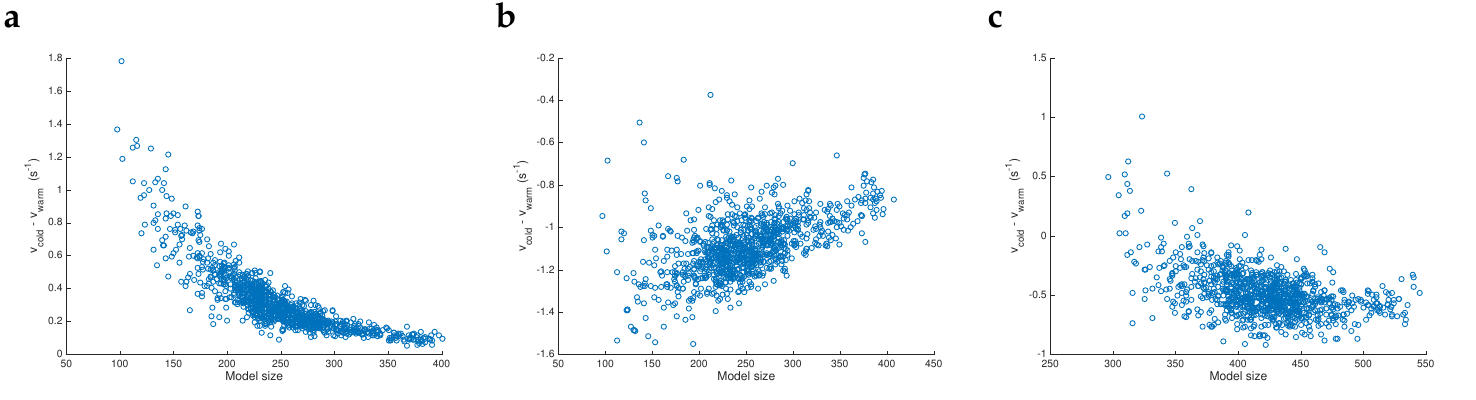}
  \caption{Difference in speed (runtime$^{-1}$) vs. model size between
    cold start ($v_{\text{cold}}$) and warm start ($v_{\text{warm}}$)
    on sub-sampled datasets for the natural coordinate descent
    algorithm (\textbf{a}), logistic regression using \texttt{glmnet}
    (\textbf{b}) and linear regression using \texttt{glmnet}
    (\textbf{c}). See main text for details.}
  \label{fig:speed}
\end{figure}

If this reasoning is correct, then the warm-start advantage should not
be present if \texttt{glmnet} is used to solve penalised least squares
problems, since in this case there is no middle loop to be
performed. To test this hypothesis, I performed penalised linear
regression (lasso) with \texttt{glmnet} on the same sub-sampled
datasets, using the same binary response variable and the same penalty
parameter values as in the previous logistic regressions. Although the
speed difference $v_{\text{cold}}-v_{\text{warm}}$ between cold and
warm start now indeed followed a similar pattern as the natural
coordinate descent algorithm (Spearman's $\rho=-0.43$; Figure
\ref{fig:speed}c), in all but 20 datasets, warm start still
outperformed cold start. This suggests that not updating the quadratic
approximation at every step during logistic regression in
\texttt{glmnet} may explain in part why it is more sensitive to the
choice of initial vector, but additional, undocumented optimizations
of the code must be in place to explain its warm-start advantage.

\section{Conclusions}
\label{sec:conclusions}

The popularity of $\ell_1$-penalised regression as a variable
selection technique owes a great deal to the availability of highly
efficient coordinate descent algorithms. For generalised linear
models, the best existing algorithm uses a quadratic least squares
approximation where the coordinate update step can be solved
analytically as a linear soft-thresholding operation. This analytic
solution has been understood primarily as a consequence of the
quadratic nature of the problem. Here it has been shown however that
in the dual picture where the penalised optimisation problem is
expressed in terms of its Legendre transform, this soft-thresholding
mechanism is generic and a direct consequence of the presence of an
$\ell_1$-penalty term. Incorporating this analytic result in a
standard coordinate descent algorithm leads to a method that is not
only theoretically attractive and easy to implement, but also appears
to offer practical advantages compared to the existing implementations
of the quadratic-approximation algorithm. In particular it is more
robust to the choice of starting vector and therefore considerably
more efficient when it is cold-started, i.e.  when a solution is
computed at set values of the $\ell_1$-penalty parameter as opposed to
along a regularisation path of descending $\ell_1$-penalties.  This
can be exploited for instance in situations where prior knowledge or
other constraints dictate the choice of $\ell_1$-penalty parameter or
in data-intensive problems where distributing the computations for
sweeping the $\ell_1$-penalty parameter space over multiple processors
can lead to significant gains in computing time. Future work will
focus on developing such parallellized implementations of the natural
coordinate descent algorithm and providing implementations for
additional commonly used generalised linear models.

\appendix

\section{Technical proofs}

\subsection{Proof of Theorem \ref{thm:main}}
\label{sec:proof-theorem-app}

With the notations introduced in Section \ref{sec:introduction}, let
$F(\beta) = U(\beta)-w^T\beta$ and
$G(\beta)=\sum_{j=1}^p\mu_j|\beta_j|$.  $F$ and $G$ are convex
functions on $\R^p$ which satisfy Fenchel's duality theorem
\citep{rockafellar1997convex}
\begin{equation}\label{eq:4}
  \min_{\beta\in\R^p} \bigl\{ F(\beta) + G(\beta)\bigr\}
  = \max_{u\in\R^p} \bigl\{ - F^*(u) - G^*(-u)\bigr\},
\end{equation}
where $F^*$ and $G^*$ are the Legendre transforms of $F$ and $G$
respectively.  We have
$F^*(u)=\max_{\beta\in\R^p}\{u^T\beta-F(\beta)\} =
\max_{\beta\in\R^p}\{(u+w)^T\beta-U(\beta)\} = L(u+w)$
and
$G^*(u) = \max_{\beta\in\R^p}\{u^T\beta-G(\beta)\} = \sum_{j=1}^p
\max_{b\in\R} \{u_j b - \mu_j|b|\}$.
If $|u_j|\leq\mu_j$ then the $j^{\text{th}}$ term in this sum is $0$,
otherwise it is $\infty$, i.e.  $G^*(u)=0$ if $|u|\lev\mu$ and
$G^*(u)=\infty$ otherwise. It follows that
\begin{equation}\label{eq:6}
  \min_{\beta\in\R^p} H(\beta) = -\min_{\{u\in\R^p\colon
    |u-w|\lev\mu\}} L(u) = - L(\hat u) = \min_{\beta\in\R^p} \bigl\{
  U(\beta) - \hat u^T\beta\bigr\},
\end{equation}
where $\hat u =\argmin_{\{u\in\R^p\colon |u-w|\lev\mu\}} L(u)$.
Denoting $\Se=\{u\in\R^p\colon |u-w|\lev\mu\}$, the minimiser $\hat u$
must satisfy the optimality conditions \citep[\S
4.2.3]{boyd2004convex}: $\hat u\in \Se$ and
\begin{equation}\label{eq:8}
  (v-\hat u)^T \grad L(\hat u) \geq 0 \quad \text{for all } v\in \Se.
\end{equation}
For any index $j$, choose $v_j\in \Se_j=\{u\in\R\colon
|w_j-u|\leq\mu_j\}$ arbitrary and set $v_k=\hat u_k$ for $k\neq
j$. Then $v\in \Se$ and by eq. \eqref{eq:8},
\begin{equation}\label{eq:15}
  (v_j-\hat u_j) \frac{\partial L}{\partial u_j} (\hat u)\geq 0.
\end{equation}
Assume $\frac{\partial L}{\partial u_j} (\hat u)\neq 0$ and
$\hat u_j\neq w_j-\sigma_j\mu_j$, where
$\sigma_j=\sgn(\frac{\partial L}{\partial u_j} (\hat u))$. Then there
exists $\epsilon>0$ such that
$v_j=\hat u_j-\epsilon\sigma_j\in \Se_j$, but this contradicts
eq. \eqref{eq:15}. By Lemma \ref{lem:legendre} below, if
$\hat\beta_0 = \argmin_{\beta\in\R^p} \bigl\{ U(\beta) - \hat
u^T\beta\bigr\}$,
then $\hat\beta_0 = \grad L(\hat u)$. Hence we have shown that
\begin{align*}
  \hat\beta_{0,j} \neq 0 \Leftrightarrow \hat u_j= w_j-\sgn(\hat\beta_{0,j})\mu_j.
\end{align*}
Denote $I=\{j\colon \hat\beta_{0,j} \neq 0\}$. We find
\begin{align*}
  \hat u^T\hat\beta_0 &= \sum_{j\in I} \hat u_j \hat \beta_{0,j} =
  \sum_{j\in I} \bigl[w_j-\sgn(\hat\beta_{0,j})\mu_j\bigr ] \hat
  \beta_{0,j} = w^T\hat\beta_0 - \sum_{j=1}^p \mu_j |\hat\beta_{0,j}|,
\end{align*}
and hence by eq. \eqref{eq:6},
\begin{align*}
  \min_{\beta\in\R^p} H(\beta) &=   U(\hat\beta_0) - \hat
  u^T\hat\beta_0 = H(\hat\beta_0),
\end{align*}
i.e. $\hat\beta_0$ is also the unique minimiser of the penalised cost
function $H$. This concludes the proof of Theorem \ref{thm:main}.

\qed

\begin{lem}\label{lem:legendre}
  For all $w\in\R^p$, we have 
  \begin{align*}
    \hat\beta_0(w) = \argmax_{\beta\in\R^p} \Bigl\{ w^T\beta -  U(\beta)
    \Bigr\} = \grad L(w).
  \end{align*}
\end{lem}
\begin{proof}
  For a given $\hat w\in\R^p$, let
\begin{align*}
  \hat\beta = \argmax_{\beta\in\R^p} \Bigl\{ \hat w^T\beta -  U(\beta)
  \Bigr\},
\end{align*}
or
\begin{align*}
  L(\hat w) = \max_{\beta\in\R^p} \Bigl\{ \hat w^T\beta -  U(\beta)  \Bigr\} =
  \hat w^T\hat\beta -  U(\hat\beta).  
\end{align*}
From Fenchel's inequality ($v^T\beta \leq U(\beta)+L(v)$ for all
$v,\beta\in\R^p$, cf. \citep[\S 3.3.2]{boyd2004convex}) it follows that
\begin{align*}
  v^T\hat\beta - L(v) \leq U(\hat\beta) = \hat w^T\hat\beta - L(\hat w),
\end{align*}
i.e.
\begin{align*}
  \hat w = \argmax_{w\in\R^p} \Bigl\{ \hat\beta^T w - L(w)\Bigr\}.
\end{align*}
By assumption $U$ is differentiable and hence so is $L$. It follows
that (see \citep[\S 3.3.2]{boyd2004convex})
\begin{align*}
  \hat \beta = \grad L(\hat w).
\end{align*}
\end{proof}

\subsection{Proof of Proposition \ref{prop:threshold}}
\label{sec:proof-prop}

First, assume $|w-w_0|\leq \mu$. By Theorem \ref{thm:main} we have
$\hat\beta=0$, and hence $\tilde w_0=w_0$ and
$|w-\tilde w_0|=|w-w_0|\leq\mu$. This establishes the direction
$|w-\tilde w_0|> \mu \Rightarrow |w-w_0|>\mu$. If $|w-w_0|> \mu$, then
again by Theorem \ref{thm:main} and using the notation from Section
\ref{sec:natur-coord-desc},
\begin{align*}
  \hat\beta = \hat\beta_0(w-\sigma\mu) = \argmin_{\beta\in\R}\bigl\{ U(\beta)
  - (w-\sigma\mu) \beta\}.
\end{align*}
Hence $U'(\hat\beta) = w-\sigma\mu$ and $\tilde w_0 =
w-\sigma\mu-U''(\hat\beta) \hat\beta$. If $\sigma=1$, $\hat\beta>0$
and, by convexity of U,
\begin{align*}
  w-\tilde w_0 = \mu + U''(\hat\beta) \hat\beta > \mu.
\end{align*}
Similarly, if $\sigma=-1$, we have $w-\tilde w_0<-\mu$. This
establishes the direction
$|w- w_0|> \mu \Rightarrow |w-\tilde w_0|> \mu$.

\qed

\section{The Cancer Genome Atlas data processing details}
\label{sec:cancer-genome-atlas}

\subsection{Breast cancer data (BRCA)}
\label{sec:breast-cancer-data}

Processed data files were obtained from
\url{https://tcga-data.nci.nih.gov/docs/publications/brca_2012/}:
\begin{itemize}
\item Normalised expression data for 17,814 genes in 547 breast cancer
  samples (file \texttt{BRCA.exp.547.med.txt}).
\item Clinical data for 850 breast cancer samples (file
  \texttt{BRCA\_Clinical.tar.gz}).
\end{itemize}
540 samples common to both files had an estrogen receptor status
reported as positive or negative in the clinical data.  Estrogen
receptor status was used as the binary response data $Y\in\R^n$,
$n=540$, and gene expression for all genes ($+$ one constant
predictor) was used as predictor data $X\in\R^{n\times p}$,
$p=17,815$.

\subsection{Colon and rectal cancer data (COAD)}
\label{sec:colon-cancer-data}

Processed data files were obtained from
\url{https://tcga-data.nci.nih.gov/docs/publications/coadread_2012/}:
\begin{itemize}
\item Normalised expression data for 20,531 genes in 270 colon and
  rectal cancer samples (file \texttt{crc\_270\_gene\_rpkm\_datac.txt}).
\item Clinical data for 276 colon and rectal cancer samples (file
  \texttt{crc\_clinical\_sheet.txt}).
\end{itemize}
266 samples common to both files had a tumor stage (from I to IV)
reported in the clinical data.  Early (I--II) and late (III--IV)
stages were grouped and used as the binary response data $Y\in\R^n$,
$n=266$, and gene expression for all genes ($+$ one constant
predictor) was used as predictor data $X\in\R^{n\times p}$,
$p=20,532$.

\section{Natural coordinate descent algorithm pseudocode}
\label{sec:code}

\begin{algorithm}[h!]
  \caption{Main loop}
  \label{alg:ncds}
  \begin{algorithmic}
    \State Initialise $\hat\beta=0$.
    \State \Call{CompleteCycle}{}
    \While{not converged}
    \While{not converged}
    \State\Call{ActiveSetCycle}{}
    \EndWhile
    \State\Call{CompleteCycle}{}
    \EndWhile
 \end{algorithmic}
\end{algorithm}
\begin{algorithm}[h!]
  \caption{Complete coordinate descent cycle}
  \begin{algorithmic}
    \Procedure{CompleteCycle}{}
    \For{$j=1,\dots,p$}
    \State\Call{CoordinateUpdate}{$j$}
    \EndFor
    \EndProcedure
\end{algorithmic}
\end{algorithm}
\begin{algorithm}[h!]
  \caption{Active set coordinate descent cycle}
  \begin{algorithmic}
    \Procedure{ActiveSetCycle}{}
    \For{$j=1,\dots,p$}
    \If{$\hat\beta_j\neq 0$}
    \State\Call{CoordinateUpdate}{j}
    \EndIf
    \EndFor
    \EndProcedure
\end{algorithmic}
\end{algorithm}
\begin{algorithm}[h!]
  \caption{Exact coordinate update}
  \begin{algorithmic}
    \Procedure{CoordinateUpdate}{j}
    \State Update $w_{0,j} = U_j'(0)$.
    \If {$|w_j - w_{0,j}| > \mu_j$}
    \State $\hat\beta_j \leftarrow$ zero of $U_j'(\cdot) - w_j + \sgn(w_j-w_{0,j})\mu_j$.
    \Else
    \State $\hat\beta_j\leftarrow 0$.
    \EndIf
    \EndProcedure
  \end{algorithmic}
\end{algorithm}
\begin{algorithm}[h!]
  \caption{Linear approximation for coordinate update}
  \label{alg:ncds-2}
  \begin{algorithmic}
    \Procedure{CoordinateUpdateLinear}{j}
    \State Update $w_{0,j} = U_j'(0)$.
    \If {$|w_j - w_{0,j}| > \mu_j$}
    \State $\hat\beta_j \leftarrow \hat\beta_j + \frac{w_j-\sgn(w_j-w_{0,j})\mu_j -
    U_j'(\hat\beta_j)}{U_j''(\hat\beta_j)}$
    \Else
    \State $\hat\beta_j\leftarrow 0$.
    \EndIf
    \EndProcedure
\end{algorithmic}
\end{algorithm}

\newpage


\end{document}